\documentclass[]{llncs}
\usepackage[dvips]{graphicx}
\usepackage{amsmath}

\newcommand{\jsim}{\mathit{sim}}

\title{Micro-Clustering: Finding Small Clusters in Large Diversity}

\author{Takeaki Uno$^1$, Hiroki Maegawa, Takanobu Nakahara$^2$, Yukinobu Hamuro$^3$, Ryo Yoshinaka$^4$, Makoto Tatsuta$^1$}
\institute{
\email{uno@nii.jp}, National Institute of Informatics, Japan
\and
Senshu University, Japan
\and
Kwansei Gakuin University, Japan
\and
Kyoto University, Japan
}

\begin{document}
\maketitle


\noindent
\begin{abstract}
We address the problem of un-supervised soft-clustering that
 we call {\em micro-clustering}.
The aim of the problem is to enumerate all groups composed of records
 strongly related to each other, while standard clustering methods
 find boundaries at where records are few.
The existing methods have several weak points; generating intractable amount
 of clusters, quite biased size distribution, robustness, etc.
We propose a new methodology {\em data polishing}.
Data polishing clarifies the cluster structures in the data by
 perturbating the data according to feasible hypothesis.
In detail, for graph clustering problems, data polishing replaces
 dense subgraphs that would correspond to clusters by cliques, and
  deletes edges not included in any dense subgraph.
The clusters are clarified as maximal cliques, thus easy to be found,
 and the number of maximal cliques is reduced to tractable numbers.
We also propose an efficient algorithm so that the computation
 is done in few minutes even for large scale data.
The computational experiments demonstrate the efficiency of our 
 formulation and algorithm, i.e., the number of solutions is small, such 
 as 1,000, the members of each group is deeply related, and the computation
  time is short.
\end{abstract}

\section{Introduction}

Unsupervised clustering is one of the central tasks in data analysis.
It is used in various areas and has many applications, such as Web access
 log analysis, recommendation, herding, and mobility analysis.
It is actually a data mining task since we can hopefully find some groups
 that are not known, or not easy to explain from background knowledge.
Usually, unsupervised clustering algorithms aim to partition the data
 into several groups, or to obtain distributions of group members.
However, this approach is often not efficient for big data with large
 diversity, since many elements may belong to several or no groups.
Another approach is to find groups of elements that are deeply related
 to each other.
Examples are community mining in social networks and biclustering by
 pattern mining.
We here call the groups of elements deeply related to each other
 {\em micro-clusters},
 and the problem of mining micro-clusters {\em micro-clustering}.

Modeling and problem formulation of micro-clustering are non-trivial.
For example, consider a similarity graph of a data in which the 
 vertices are elements, and two vertices are connected if the
 corresponding two elements are similar, or strongly related.
A clique (complete subgraph) of a similarity graph is a good
 candidate for micro-clustering, thus we are motivated to enumerate all
 maximal cliques in the graph, where a maximal clique is a clique
 included in no other clique\footnote{maximal clique and maximum cliques
 are different; a maximum clique is clique of the maximum size.}.
Community mining often uses maximal clique enumeration in social networks
 in which the cliques are considered as cores of the communities.
However, similarity graphs usually include a huge number of maximal
 cliques due to ambiguities, thus maximal cliques enumeration is often
 intractable in practice.
Some propagation methods\cite{BP,HITS} extract structures that can
 be seen as micro-clusters.
Then, completeness becomes a difficult issue.
Repetitive executions of these methods yield many similar groups
 corresponding to one group, and often miss some weak groups that are 
 hidden by intensive groups.
Global optimization approach such as modularity maximization\cite{Newman}
 or graph cut\cite{Metis}, 
To the best of our knowledge, there is no efficient method for
 micro-clustering.

One of the difficulty of micro-clustering is the problem formulation;
 the clusters in the data is hard to describe only with mathematical terms.
``Enumeration of maximal cliques'' involves numerous solutions, ``modularity
 maximization'' gives incentives to enlarge large clusters much more, 
 ``finding local dense structures by propagation (random walks)'' hides 
 small weak clusters surrounding a dense and relatively large cluster.
The form of clusters differ in many kinds of data and problems, and our 
 objectives, so it is naturally difficult, and there may be no general 
 answer.
On the other hand, we have known the result of the clustering algorithms
 are often good, and valuable to use for machine learning, visualization, 
 natural language processing, etc.
This fact implies that computing good clusters might be easier than 
 modeling the clusters in general terms.
According to the Vapnik's principle, now we might solve a difficult problem
 to solve a relatively easier problem, that is said to be not a good approach.
It is worth considering how to obtain good clusters without mathematically
 written problem formulation.

In this paper, instead of solving clustering problem directly,  we
 address the problem of eliminating the ambiguity to make the clustering easy.
We propose a new concept called {\em data polishing} that is to clarify
 the structures in data to be found by modifying the data according
  to feasible hypothesis.
For example, it modifies the input graph so that similar structures
 (groups, patterns, segments, etc.) considered to be the same one
 (having the same meaning) will be unified into one.
The goal of the data polishing is that we will have one to one
 correspondence between the objects (meanings) we want to find,
  and the structures modeling the objects.
The difference from data cleaning is that data cleaning usually correct the 
 each entity of the data to delete the noise, thus does not aim to change
 the solutions to the original data, while data polishing actively
 modifies the data so that the solutions will change.
For example, consider an example of community mining by finding
 cliques of size $k$.
A typical data cleaning deletes vertices of degrees less than $k-1$ so
 that no clique of size $k$ will be missed.
In a data polishing way, we consider that a community is a pseudo clique
 (dense subgraph) in the graph, thus replace pseudo cliques by cliques.
This drastically changes the result of the clique mining.
A community in an SNS network is a dense subgraph, but it usually has
 many missing edges, thus it is not a clique.
These missing edges are actually ``truth'' and not noise, thus they are
 difficult to recognize them as pairs that should be connected,
 with the data cleaning analogy.
However, a community is preferred to be a clique, in the sense of
 the motivation of modeling of communities.
Data polishing modifies the data so that the structures we want to find
 will be what we think they should be, instead of allowing 
 the individual record to be broken.
We observe that when the graph is clear such that it is composed only
 of cliques corresponding to communities, the mining results will be fine.
This is the origin of the idea of data polishing.

The enumeration of pseudo cliques is actually a hard task; graphs often
 include much more maximal pseudo cliques compared to maximal
  cliques\cite{pce}.
Instead of enumerating pseudo cliques, we use a feasible hypothesis;
 two vertices have many common neighbors in the graph if they are
  included in a dense subgraph of a certain size
  (see Figure \ref{dense}).
We can also use a similarity of neighbors instead of the number of
 common neighbors.
It is possible that two vertices are included in a dense subgraph 
 but have few common neighbors for example a graph composed of
  a vertex $v$ and a clique.
In such cases, $v$ and the clique should not be in the same
 micro-cluster, thus we think this hypothesis is acceptable.
We find all vertex pairs having at least a certain number of common
 neighbors (or the similarity of neighbors is no less than the
  given threshold), and connect each pair by an edge.
On the other hand, we delete all edges whose endpoints do not satisfy the
 condition, since they are considered as not being in the same cluster.
We repeat this operation until the graph will not change, and obtain
 a ``polished'' graph.

In summary, our micro-clustering algorithm is composed of three parts,
 (1) construction of the similarity graph, (2) data polishing,
 and (3) maximal clique enumeration.
As the number of maximal cliques becomes small through data polishing, 
 (3) is not a difficult task.
The computational cost of (1) and (2) depends on the choice of the
 similarity measure and the definition of data polishing.
Usually these tasks are not heavy because the data and graph are usually 
 sparse.
Our computational experiments show the efficiency of our data polishing
 algorithm.
For example, it reduces the number of maximal cliques in a social
 network from 33,000 to 300 and in a similarity graph of news articles 
 from over 55 million to 100,000.
The contents of the cluster are acceptable and understandable, at least 
 in our computational experiments.
The detection accuracy of clusters for randomly generated data is
 significantly better compared to existing algorithms.
The quality and the performance of our algorithm drastically outperforms 
 the existing algorithms.

\begin{figure}[t]
  \begin{center}
  \includegraphics[scale=0.45]{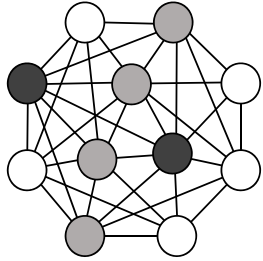}
  \end{center}
  \caption{Two vertices (black) in a dense graph have many (4) common neighbors (gray)}\label{dense}
\end{figure}

The organization of this paper is as follows.
The next section introduces the notations.
In Section 3, we describe our model of data polishing for
 micro-clustering in a graph, and describe a fast algorithm for data
 polishing and similarity graph construction in Section 4.
Section 5 explains an algorithm for maximal clique enumeration.
In Section 6, we show the results of computational results, and discuss
 the result in Section 7.
We conclude the paper in Section 8.

\subsection{Related Works}

There are many unsupervised clustering algorithms.
Typical algorithms partition data into two or several
 groups by learning the boundaries of the groups.
A typical approach to obtain many clusters is to recursively apply these
 clustering algorithms.
However, in the top levels of the partition, finding a boundary composed
 of boundaries of many clusters is very difficult, and micro-clusters
 are often broken by global cuts.

For finding clusters in a similarity graph, Girvan-Newman clustering
 algorithm\cite{Newman} and its many variants are often used.
The algorithms partition the data into many groups according to the
 {\em modularity}, i.e., a model of clustering quality.
However, they often produce few very large groups that are a
 mixture of many micro-clusters and many quite small groups.
We observed this in our computational experiments.
 
Some random walk algorithms and propagation algorithms\cite{BP,HITS}
 find a cluster from a seed vertex.
The obtained cluster can be seen as a main cluster among those to which
 the seed vertex belongs.
For enumeration, we often apply these algorithms to all vertices
 individually.
This results in many quite similar groups that correspond to one cluster;
 on the other hand, weak clusters hidden by strong clusters may be missed.
Moreover, even one execution of these algorithms takes sufficiently
 long time, the total computation time becomes huge.
 
Some researchers attempted to enumerate maximal cliques from graphs.
For example, Kumar et al. enumerated bipartite maximal cliques from
 Web link graphs\cite{KmRhRjTk99}.
However, the number of solutions is usually so huge.
In our experiments, the number of maximal cliques in a similarity graph
 from Reuters news articles was above 50,000,000.
There are currently no algorithm to drastically reduce the number
 of maximal cliques.

An idea to remove edges from the network so that extraction of
 clusters will become easy has been proposed by Satuluri et. al.
 \cite{sparsify}, that is called {\em local sparsification}.
The idea to identify the edges to be removed from the network is 
 same as our method, but they do not add edges, and do not repeat
 the process, thus the result would be quite different.
Also a similar idea exists that is called stochastic flow clustering
 method\cite{sfc}.
The idea is update distance matrix iteratively so that the items belonging 
to the same cluster would have distance of zero.
Compared to this, our method can be considered as a discrete version
 with allowing scalable algorithms applicable to large scale graph data.

\section{Preliminary}

A {\em database} $D$ is a collection of {\em records} $T_1,\ldots,T_n$.
A record can be any kinds of structures, but we restrict the database 
 such that each record belongs to the same class of structures, such as
 itemsets, sequences, and labeled graphs.
For a similarity measure $sim$, the {\em similarity graph} $G_{sim}(D)$ 
 of the database $D$ is the graph such that the vertex set is its records,
 and an edge connects two records $u$ and $v$ if and only if
  $sim(u,v) \ge \theta$.
We simply write the similarity graph $G=(V,E)$ and assume that
 $V=\{1,...,n\}$ if there is no confusion.
A {\em clique} is a set of vertices such that every pair of vertices
 is connected by an edge.
Cliques are usually defined by a subgraph, but in this paper we 
 use this definition for conciseness.
A clique included in no other clique is called {\em maximal}.
The size of a clique is the number of its vertices.

A pair of vertices not connected by an edge is called {\em non-edge}.
The {\em density} of a graph $G=(V,E)$ is the ratio of edge existences,
 i.e., $\frac{|E|}{|V|(|V|-1)/2}$.
The density of a vertex set $U$ is defined by the density of the subgraph
 $G[U] = (U, E[U])$, where $E[U]$ is the set of edges that are connecting
 two vertices in $U$.
A {\em pseudo clique} is a vertex set having sufficiently large density,
 such as that with density no less than the threshold $\delta$.
We call a vertex set of density $\delta$ {\em $\delta$-pseudo clique}.

For vertex $v$ in $G$, a vertex adjacent to $v$ is called a
 {\em neighbor} of $v$.
The set of neighbors of $v$ is denoted by $N(v)$.
$N[v]$ denotes $N(v)\cup \{ v\}$ and is called the {\em closed neighbor}.
A vertex $w$ is a {\em common neighbor} of vertices $u$ and $v$ if 
 $w$ is adjacent to both $u$ and $v$. 
The {\em degree} of a vertex $v$ is the number of vertices adjacent
 to $v$, and denoted by $d(v)$.
For a subgraph $K$ of $G$, the degree of $v$ in $K$ is the number of
 vertices in $K$ that are adjacent to $v$ and denoted by $d_K(v)$.
The minimum degree in $K$ is $\min_{v\in K} d_K(v)$.

\section{Data Polishing for Micro-Clustering}

We first state our micro-clustering problem from conceptual aspects, and 
 propose our data polishing algorithms to clarify the clusters.
We also state some properties that are mathematical aspects of our
 algorithms.
Note that the statements of properties are of the worst cases, thus
 they do not mean the average cases.
Micro-clusters are groups of data records that are similar or related
 to each other, and desired to have one meaning, or correspond to 
 one group.
By considering the applications of micro-clustering, 
 micro-clusters should satisfy the following conditions.

\begin{itemize}
\item[1] quantity (the number of micro-clusters found should not be huge)
\item[2] independence (micro-clusters should not be similar to each other)
\item[3] coverage (all micro-clusters should be found)
\item[4] granularity (the granularity of micro-clusters should be the same)
\item[5] rigidity (the micro-clusters found should not change due to
 un-essential changes such as random seeds or indices of records)
\end{itemize}

From conditions 3 and 4, we approach the micro-clustering problem by
 structural enumeration, since structures such as trees and cliques
 are considered to correspond micro-clusters with similar granularity.
Condition 5 leads us not to use randomness and computation depends 
 on the vertex ID ordering.
For the remaining conditions 1 and 2, we can observe that the ambiguity
 in the data makes algorithms violating the conditions; noise
 drastically increase the structures to be enumerated, and the
 structures are too similar thereby not independent.
In a graph, micro-clusters are considered to correspond to dense subgraphs,
 and the non-edges in the dense subgraphs are ambiguity.
We also consider that edges included in no clusters are also ambiguity. 
The concept of data polishing for micro-clustering comes from these;
 add edges for these non-edges, and remove these edges from the graph.
For identifying these non-edges and edges, we consider the following
 feasible hypothesis.

If vertices $u$ and $v$ are in the same clique of size $k$, $u$ and $v$
 have at least $k-2$ common neighbors.
Thus, we have $|N[u]\cap N[v]| \ge k$, and this is a necessary condition 
 that $u$ and $v$ are in a clique of size at least $k$.
We call this condition {\em $k$-common neighbor condition}.
If $u$ and $v$ are in a sufficiently large pseudo clique, they are also
 expected to satisfy this condition.
For example, the average degree in a pseudo clique of size $2k$ whose edge
 density is $80\%$ is $1.6k$, thus any of its two vertices satisfy the
  condition with high probability.
On the contrary, if two vertices do not satisfy the condition, they
 belong to a pseudo clique with very small probability.
Even though they belong to a pseudo clique, they actually seem to be
 disconnected in the clique, thus we may consider they should not be 
 in the same cluster.
Let $P^k(G) = (V, E')$ where $E'$ is the correction of edges connecting
 vertex pairs satisfying the $k$-common neighbor condition, and the 
 polishing process is the computation of $P^k(G)$ from $G$.
We call this process {\em $k$-intersection polishing}.
 
The following property assures that we can find pseudo cliques as cliques
 in the processed graph.

\begin{property}
A vertex set $K\subseteq V$ of $\gamma k$ vertices is a clique in 
 $P^k(G)$ when the minimum degree of $K$ in $G$ is at least $(\gamma + 1)k/2$.
\end{property}

\proof
Any two vertices $u$ and $v$ can have at most
$2 \times (\gamma k-1 - \frac{(\gamma + 1)k}{2}) = \gamma k -2 - k$
 non-common neighbors.
Thus, $|N[u]\cap N[v]| \ge k$.
\qed

Note that the density of $K$ is always no less than $(\gamma + 1)/2\gamma$.

\begin{theorem}
Let $K$ be a $\delta$-pseudo clique ($\delta < 1$) in $G$ of size $\gamma k$
 where the degree of each node is at least $(\gamma+1)k/s-1$ for
 $\gamma > 1$ and $s > 2$.
The density of $K$ in $P^k(G)$ is more than that in $G$
 if $1-\delta \le \frac{(s-1)(\gamma-1)^2}{s^2\gamma^2}$.
\end{theorem}

\proof
Partition the vertices of $K$ into
\begin{eqnarray*}
 P &=& \{\, v \in V \mid (\gamma+1)k/s \le d_K(v) < (\gamma+1)k/2 \,\},\\
 Q &=& \{\, v \in V \mid (\gamma+1)k/2 \le d_K(v) < (s-1)(\gamma+1)k/s \,\},
  \mbox{and}\\
 R &=& \{\, v \in V \mid (s-1)(\gamma+1)k/s \le d_K(v) \,\}\,.
\end{eqnarray*}
Let $p = |P|$ and $q = |Q|$.
Concerning the numbers $M$ and $M'$ of non-edges in $K$ and $K$ in $P^k(G)$,
 respectively, we have the following:
\begin{eqnarray}
	M &=& (1-\delta)\frac{\gamma k(\gamma k - 1)}{2} < (1-\delta)\frac{\gamma^2k^2}{2} 
	 \label{M1}\\
M &>& \frac{1}{2}\left({\left(\gamma k - \frac{(\gamma+1)k}{2}\right)p +
 \left(\gamma k - \frac{(s-1)(\gamma+1)k}{s}\right)q}\right)\\
 &=& \left(\frac{p}{4}+\frac{q}{2s}\right)(\gamma-1)k\,,
\label{M2}
\\	M' &\le& pq+\frac{p(p-1)}{2} \le p\left(q+\frac{p}{2}\right)\,.
\label{M3}
\end{eqnarray}
From Eq.~(\ref{M2}), we have $q < \frac{2sM}{(\gamma - 1) k} - \frac{s}{2}p$.
Together with Eq. ~(\ref{M3}), 
\begin{eqnarray*}
M' &<& p\left( \frac{2sM}{(\gamma - 1) k} - \frac{s-1}{2}p \right)
\\ &=& - \frac{s-1}{2}\left(p - \frac{2sM}{(s-1)(\gamma - 1) k} \right)^2 + \frac{2(sM)^2}{(s-1)((\gamma - 1) k)^2}
\\	&\le& \frac{2s^2M^2}{(s-1)(\gamma - 1)^2 k^2}\,.
\end{eqnarray*}
Therefore, we have
$\frac{M'}{M} < \frac{2s^2M}{(s-1)(\gamma - 1)^2 k^2}
 < \frac{s^2 (1-\delta)\gamma^2 k^2}{(s-1)(\gamma - 1)^2 k^2} \le 1$
from Eq.~(\ref{M1}) and $(1-\delta) \le \frac{(s-1)(\gamma-1)^2}{s^2\gamma^2}$.
\qed

\begin{corollary}
Let $K$ be a $8/9$-pseudo clique of $\gamma k$ vertices whose minimum
 degree is at least $(\gamma+1)k/3$.
The density of $K$ always increases by $k$-intersection polishing if 
 $\gamma > 2+2^{1/2} \simeq 3.4142$.
\end{corollary}

From the theorem and corollary, we see that any sufficiently large 
 pseudo clique with sufficiently large density increases its density
 in $P^k(G)$.
Since the condition of the theorem is far from tight, we can expect that
 the densities of smaller subgraphs also increase in real world data.
Hence, by repeatedly computing $P^k(G)$, $P^k(P^k(G))$ and so on, 
 we hope that it will converge to a graph $P^{k*}(G)$ satisfying
 $P^{k*}(G) = P^k(P^{k*}(G))$.
However, there is a counter example, thus we choose a
 number $\tau$ and repeat the computation at most $\tau$ times.

\begin{figure}[t]
  \begin{center}
  \includegraphics[scale=0.45]{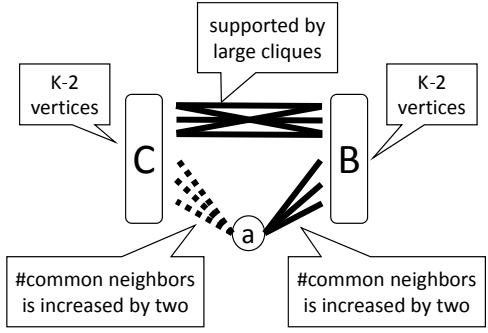}
  \end{center}
  \caption{Example of a graph that $k$ intersection polishing does not converge}\label{non-converge}
\end{figure}
 
\begin{theorem}
There is a graph $G'=(V',E')$ such that $P^k(G'') \ne G''$ holds 
 for any $G'' = P^k(P^k(\cdots P^k(G')\cdots))$ for $k>4$.
\end{theorem}

\proof 
We describe the example shown in Figure \ref{non-converge}.
In the figure, $B$ and $C$ are a set of $k-2$ vertices such that
 vertex $(\{a\},B)$ and $(B,C)$ induce complete bipartite graphs where 
 no vertex pair in $B$ or $C$ is connected.
For each $u\in B$ and $v\in C$, we put a clique of size $k+2$
 so that the clique includes $(u,v)$ but no other vertex in the graph.
This is to make $u$ and $v$ always have at least $k+2$ common neighbors.
For each $v\in B\cup C$, we prepare vertices $w_1$ and $w_2$, and
 put two cliques of size $k+2$ including only $(a,w_i)$ and $(w_i,v)$
 in the same manner.
This is to increase the common neighbors of $a$ and $v$ by two.

Let us consider $|N[u]\cap N[v]|$ for each vertex pair $(u,v)$.
For each edge $(u,v)$ in an attached clique, $|N[u]\cap N[v]|\ge k+2$.
Each vertex $u\not \in \{a\}\cup B\cup C$ satisfies $|N[u]\cap N[v]|\le 2$
 for any vertex $v$ not in the clique that includes $u$.
Any vertex pair $(u,v)$ taken from $B$ or $C$ satisfies
 $|N[u]\cap N[v]| \le k-1$.
For each $v\in B$, we have $|N[a]\cap N[v]|=4$, and
 for each $v\in C$, we have $|N[a]\cap N[v]|=k$.
Hence, when we apply one iteration of $k$-intersection polishing
 with $k>4$, all edges between $a$ and $B$ disappear and all non-edges between 
 $a$ and $C$ become edges.
The edge set of the new graph is different from the original graph, but
 the new graph is isomorphic to the original.
Therefore, the statement holds.
\qed

It is worth mentioning that this statement holds if $k$-intersection 
polishing is defined by a usual neighbor but not a closed neighbor,
 i.e., $N(u)\cap N(v)$.
\qed


We call a graph $G$ {\em polished} if it does not change by data
 polishing process.
In particular, we call $G$ {\em $k$-intersection polished} if $G$
 satisfies $G = P^k(G)$.

\begin{theorem}
A $k$-intersection polished graph $G=(V,E)$ has at most $O(|V|^k)$
 maximal cliques.
\end{theorem}

\proof
Suppose that a vertex set $S, |S|=k$ is included in two distinct
 cliques $K$ and $K'$.
Then, for any two vertices $u\in K$ and $v\in K'$, $S\subseteq N[v], N[u]$
 hold, thus $|N[u]\cap N[v]| \ge k$.
This implies that such $u$ and $v$ are always connected by an edge, and 
 thus $K\cup K'$ is a clique.
Hence, we can observe that any vertex subset $S\subseteq V$ of size $k$
 is included in at most one maximal clique.
Since any maximal clique of size at least $k$ includes at least one vertex 
 subset of size $k$, we can see that the number of maximal cliques of size
  at least $k$ is at most $_nC_k$.
The number of cliques of size at most $k-1$ is bounded by $n^k$, the number 
 of maximal cliques in $G$ is bounded by $_nC_k + n^k = O(n^k)$.
Therefore, the statement holds.
\qed

This theorem demonstrates that $k$-intersection polishing reduces
 the number of maximal cliques.
Moreover, a limited number of maximal cliques makes some problems easy.
For example, we can say that maximal clique enumeration in $G$ can be solved
 in polynomial time in the size of $G$, since the enumeration is done 
 in linear time in the number of maximal cliques\cite{mace}.
We can also observe a tractability on the maximum clique problem.

\begin{corollary}
The maximum clique problem in a graph $G$ polished by $k$-intersection
 polishing can be solved in polynomial time, $O(n^{k+2.376})$ time in
 particular.
\end{corollary}

\proof
The algorithm MACE\cite{mace} enumerates all maximal cliques in time 
 $O(Mn^{2.376})$, where $M$ is the number of maximal cliques in $G$.
This together with the above theorem leads to finding a maximum
 clique in $O(n^{k+2.376})$ time.
\qed

However, we can observe that some polished graphs have a huge
 number of maximal cliques for large $k$.
The following property demonstrates that there are some graphs 
 having exponentially many maximal cliques.

\begin{theorem}
There is a series of infinitely many $k$-intersection polished graphs that
 have exponentially many maximal cliques.
In particular, a graph of $n^2/4 + 3n/2$ vertices has at least $2^{n/2}$
 maximal cliques for some $k$.
\end{theorem}

\proof
Let $G$ be a graph of vertex set $V = \{1,\ldots,n\}$ such that 
 any vertex pair except for $2i-1$ and $2i$ for $i=1,\ldots,n/2-1$
 is connected by an edge.
$G$ is the removal of a perfect matching from a complete graph.
We assume that $n$ is a multiplier of 2.
The graph has exactly $2^{n/2}$ maximal cliques, since any choice
 of one vertex from each non-edge forms a maximal clique.
Any $u$ and $v$ satisfies $|N[u]\cap N[v]| = n-2$ in $G$.
We add vertices and edges to $G$ to obtain a graph $G'$ satisfying 
 the statement.

Let $V_1 = \{1,3,5,\ldots,n-1\}$ and $V_2 = \{2,4,6,\ldots,n\}$ be
 a partition of $V$ into two subsets of equal size.
We then consider sets of vertices $U_0 = V_2$, and $U_i = V_1\setminus
 \{2i-1\} \cup \{2i\}$, for $1\le i\le n/2$.
We can see that any edge in $G$ is included in at least one of $U_j, 
 0\le j\le n/2$, and no non-edge is included in some $U_j$.
For each $U_j$, we add $n/2$ vertices and connect edges so that 
 $U_j$ and the added vertices form a clique.
The obtained graph is denoted by $G' = (V', E')$.
We have $|V'| = n + (n/2+1) \times n/2 = n^2/4 + 3n/2$.
Observe that by the addition, a pair of vertices increases their common
 neighbors by $n/2$ if and only if they are connected by an edge.
Thus, the graph is an $n$-intersection polished graph.
We can see that any maximal clique in $G$ other than $U_j$ is a maximal 
 clique in $G'$, thus we still have $2^{n/2}$ maximal cliques.
Therefore, the statement holds.
\qed

For real world data, $k$-intersection polishing generates acceptable
 clusters in some cases such as when the pseudo cliques are disjoint.
However, it often does not work well if there are many vertices adjacent
 to several large degree vertices (hubs), since these vertices form
 a very large clique even though they are not related.
This property is often observed in many kinds of networks such as
 social networks, and Web link networks.
The cause is that two vertices are connected, even though many of their
 neighbors are not common.
In such cases, instead of $k$-intersection model, similarity measures
 are preferable such as the Jaccard coefficient.
By using similarity, vertices adjacent to many different vertices
 will not be connected.

We denote the similarity of $N[u]$ and $N[v]$ by $sim(u,v)$, where
 $N[u]$ and $N[v]$ are more similar when $sim(u,v)$ is large.
For example, if the similarity measure is the Jaccard coefficient, 
 $sim(u,v) = |N[u]\cup N[v]| / |N[u]\cap N[v]|$.
We define a graph $P^{sim(\theta)}(G)$ in the same manner, i.e.,
 $P^{sim(\theta)}(G) = (V, \{(u,v)| sim(u,v)\ge \theta\})$.
The operation of computing $P^{sim(\theta)}(G)$ is called $\theta-sim$
 polishing, and a graph is called $\theta-sim$ polished if
  $G=P^{sim(\theta)}(G)$.
In particular, if the similarity measure is the Jaccard coefficient, 
 they are $\theta$-Jaccard polishing and $\theta$-Jaccard polished.
Unfortunately, even in these cases polishing may not converge for
 some graphs.

\begin{theorem}
For every\/ $\theta \in (0,1/4] \cup (3/11,2/7]$, there is a graph\/ $G = (V,E)$ such that\/
$P^{\jsim(\theta)}(G) \neq G$ and\/ $P^{\jsim(\theta)}(P^{\jsim(\theta)}(G)) = G$.
\end{theorem}

\begin{proof}
We define a graph $G_{m_1,m_2,n} = (V,E)$ as follows.
Let
\[ V = \{a\} \cup B \cup C \cup \bigcup_{i,j \in \{1,\dots,n\}} D_{i,j}
\]
where $B=\{b_1,\dots,b_n\}$, $C=\{c_1,\dots,c_n\}$, $|D_{i,i}|=m_1$,
 and $|D_{i,j}|=m_2$ for $i \neq j$.
The edge set $E$ is defined so that $\{a\}\times B$ and $B\times C$ form bicliques and $D_{i,j} \cup \{b_i,c_j\}$ forms a clique for each $i,j$.
No other vertices are connected.
Then, for each $i,j \in \{1,\dots, n\}$ and $d \in D_{i,j}$,
\begin{align*}
N[a] &= \{a\} \cup B 
	& |N[a]| &= n+1,
\\
N[b_i] &= \{a,b_i\} \cup C \cup \bigcup_{j} D_{i,j},
	& |N[b_i]| &=  m_1 + (n-1)m_2+n +2,
\\
N[c_j] &= \{c_j\} \cup B \cup \bigcup_{i} D_{i,j},
	& |N[c_j]| &=  m_1 + (n-1)m_2 +n + 1,
\\
	N[d] & = D_{i,j} \cup \{b_i,c_j\},  & |N[d]| &= m+2,
\end{align*}
where $m = m_1$ if $d \in D_{i,i}$ and $m =m_2$ otherwise.
Thus, for $i' \neq i$, $j' \neq j$, $d' \in D_{i,j} \setminus\{d\}$ and $d'' \notin \{a\} \cup B \cup C \cup D_{i,j}$,
\begin{align*}
N[a] \cap N[b_i] &= \{a,b_i\},
&
\jsim(a,b) &= \frac{2}{2n+m_1 + (n-1)m_2+1},
\\
N[b_i] \cap N[b_{i'}] &= \{a\} \cup C,
&
\jsim(b_i,b_{i'}) &=  \frac{n+1}{2m_1 + 2(n-1)m_2+n +3},
\\
N[c_j] \cap N[c_{j'}] &= B,
&
\jsim(c_j,c_{j'}) &=  \frac{n}{2m_1 + 2(n-1)m_2 +n + 2}, 
\\
N[a] \cap N[c_j] &= B,
&
\jsim(a,c_j) &= \frac{n}{ m_1 + (n-1)m_2 + n+2},
\\
N[b_i] \cap N[c_j] &= D_{i,j} \cup \{b_i,c_j\},
&
\jsim(b_i,c_j) &= \frac{m+2}{2m_1 + 2(n-1)m_2 -m+2n + 1},
\\
N[b_i] \cap N[d] &= D_{i,j} \cup \{b_i,c_j\},
&
\jsim(b_i,d) &= \frac{m+2}{ m_1 + (n-1)m_2+n +2},
\\
N[c_j] \cap N[d] &= D_{i,j} \cup \{b_i,c_j\},
&
\jsim(c_j,d) &= \frac{m+2}{ m_1 + (n-1)m_2+n +1},
\\
N[d] \cap N[d'] & = D_{i,j} \cup \{b_i,c_j\} 
&
\jsim(d,d') & = 1,
\\
N[d] \cap N[d''] &\le 1
&
\jsim(d,d'') & \le \frac{1}{\min\{m_1,m_2\}+m_2+3} 
\end{align*}
and $\jsim(u,v)=0$ for any other pairs of distinct $u$ and $v$.

Now we discuss when the graph $ P^{\jsim(\theta)}(G_{m_1,m_2,n}) $ will be symmetric to $G$ in the sense that edges between $\{a\}$ and $B$ have been removed and $\{a\} \times C$ have become a biclique while the other edges are preserved.
This occurs when $f(m_1,m_2,n) < \theta \le g(m_1,m_2,n)$ for
\[
	f({m_1,m_2,n}) = \max\left\{ \jsim(a,b_i),\jsim(b_i,b_{i'}),\jsim(c_j,c_{j'}),\jsim(d,d'') \right\}
\]
 and
\[
	g({m_1,m_2,n}) = \min\left\{ \jsim(a,c_j),\jsim(b_i,c_j),\jsim(b_i,d),\jsim(c_j,d),\jsim(d,d')\right\}\,.
\]
We have
\[\begin{array}{|c|c|c|c|c|}\hline
	m_1	&	m_2	&	n	&	f(m_1,m_2,n)	&	g(m_1,m_2,n)
\\ \hline
	1	&	2	&	2	&	3/11	&	2/7
\\ \hline
	2	&	2	&	2	&	3/13 	&	1/4 
\\ \hline
	2	&	2	&	3	&	2/9	 &	4/17 (> 3/13) 
\\ \hline
	2	&	3	&	2	&	1/5	 &	2/9 
\\ \hline
	3	&	3	&	3	&	1/6	&	3/14 (> 1/5) 
\\ \hline
\end{array}\]
and for all $m \ge 3$, 
\[
	g(m,m,3) > g(m+1,m+1,3) > f(m,m,3) > f(m+1,m+1,3)
\,.\]
Note that $\lim_{m \to \infty} g(m,m,3)=0$.
Therefore, for every $\theta \in (0,1/4] \cup (3/11,2/7]$, there is a graph $G$  satisfying the property of the theorem. \qed
\end{proof}

Finally, we describe our micro-clustering algorithm and 
 data polishing algorithm.
Here $\tau$ is the number of the maximum repetitions.

\begin{tabbing}
{\bf Algorithm} Micro-Clustering\_Similarity ($D$:database, $\theta',
 \theta$:thresholds for record\\ similarity and neighbor similarity)\\
1. construct the similarity graph $G_{sim'}(D)$ with threshold $\theta'$\\
2. $G' =$ GraphPolishing\_for\_Cluster ($G_{sim'}(D)$, $\theta$)\\
3. {\bf output} all maximal cliques of $G'$
\end{tabbing}

\begin{tabbing}
{\bf Algorithm} GraphPolishing\_for\_Cluster ($G=(V,E)$:graph,
 $\theta$:threshold)\\
1. {\bf for} $i := 1$ to $\tau$\\
2. \ \ \ $E' := \{ (u,v) | sim(u,v) \ge \theta \}$\\
3. \ \ \ {\bf if} $E' = E$ {\bf then break}\\
4. {\bf end for}\\
5. {\bf output} $G' = (V,E')$
\end{tabbing}

The maximal clique enumeration is done by algorithms such as
 MACE\cite{mace}\footnote{http://research.nii.ac.jp/\~{}uno/codes.html},
 and Tomita's algorithm
 \cite{tomita}\footnote{http://research.nii.ac.jp/\~{}uno/codes.html}.
For the computation of $P^{sim(\theta)}(G)$ and similarity graph for
 set families (transaction databases), the algorithm described in the 
 next section is efficient\cite{DS04}.

\section{Algorithm for Fast Data Polishing}

A straightforward way to compute $sim(u,v)$ for all pairs of $u$ and $v$ 
 takes at least $O(|V|^2)$ time.
This is very heavy when $|V|$ is large, such as one million.
We can not avoid this difficulty when the similarity graph or the
 polished graph is dense, and thus has $\Theta(|V|^2)$ edges.
In fact, this implies that elements are similar to many others, thus
 there are only few clusters.
They can be tractable by usual clustering algorithms such as $k$-means.
Micro-clustering aims to find many small clusters, thus such dense graphs
 are not interesting, and we assume that polished graphs are sparse.
However, Even with this assumption, computation less than
 square time is non-trivial.
For efficient computation, we observe the following.\\

\noindent
{\bf Observation}: 
In many similarity measures for neighbors, $sim(u,v) \ge \theta$ only when 
 $|N[u]\cap N[v]| > 0$\\

If the similarity graph is sparse, the intersection size $|N[u]\cap N[v]|$
 is zero for almost all pairs of vertices, and only few pairs satisfy
  $|N[u]\cap N[v]| > 0$.
Moreover, many similarity measure for sets can be computed quickly
 with the use of $|N[u]\cap N[v]|$, since $|N[u]\cup N[v]| =
  |N[u]| + |N[v]| - |N[u]\cap N[v]|$, $|N[u]\setminus N[v]| =
  |N[u]| - |N[u]\cap N[v]|$, and so on.
Therefore, we describe an algorithm for computing the intersection
 of closed neighbors, for each pair having non-empty intersection.
The algorithm is actually a kind of folklore and often used in
 literature\cite{DS04}.
The algorithm can be seen as a variant of the transposition of
 a matrix represented in a sparse manner.
We additionally show a bound for computation time when the degree
 sequence of the graph satisfies the Zipf's law.
To the best of our knowledge, this is the first theoretical result that 
 explains the reason this folklore algorithm is fast in practice.
We begin with the following property.

\begin{property}
$N[u]$ has non-empty intersection with $N[v]$ 
 if and only if $u$ and $v$ are both in $N[w]$ for some $w$. \qed
\end{property}

From this property, we can see the following property.

\begin{property}
$|N[u]\cap N[v]|$ is equal to the number of vertices $w\in N[u]$ such that 
 $v\in N[w]$. \qed
\end{property}

We can see that $|N[u]\cap N[v]|$ for all $v$ can be computed by scanning
 $N[w]$ for all $w\in N[u]$, i.e., we scan $N[w]$ one by one.
For each $v\in N[w]$, we increase the counter for $v$, then the counter
 will be the intersection size after all scans.
This idea is described as the following algorithm.

\begin{tabbing}
{\bf for each} $w\in N[u]$ {\bf do}\\
\ \ {\bf for each} $v\in N[w]$ s.t. $v < u$, $intersection[v] := intersection[v]+1$\\
{\bf end for}
\end{tabbing}


In the algorithm, we want to output non-zero intersection size.
For the sake, we use a list $L$ that keeps the vertices having non-empty
 intersection with $u$.
A vertex $v$ is inserted to $L$ when $intersection[v]$ is changed from
 $0$ to $1$.
The following algorithm computes the intersection size for all pairs 
 by executing the above algorithm for all $u\in V$.

\begin{tabbing}
{\bf Algorithm} NeighborIntersection ($G=(V,E)$)\\
1. \ \= $intersection[u] := 0$ for each $u\in V$\\
2. \>{\bf for each} $u\in V$ {\bf do}\\
3. \>\ \ $L := \emptyset$\\
4. \>\ \ {\bf for each} $w\in N[u]$ {\bf do}\\
5. \>\ \ \ \ {\bf for each} $v\in N[w]$ s.t. $v < u$ {\bf do}\\
6. \>\ \ \ \ \ \ {\bf if} $intersection[v] = 0$ {\bf then} insert $v$ to $L$\\
7. \>\ \ \ \ \ \ $intersection[v] := intersection[v]+1$\\
8. \>\ \ \ \ {\bf end for}\\
9. \>\ \ {\bf end for}\\
10.\>\ \ {\bf for each} $v\in L$ {\bf output} ``$\{v,u\} intersection[v]$'',
 $intersection[v] := 0$\\
11.\> {\bf end for}
\end{tabbing}

Note that this algorithm also works for set families (transaction
 databases) where a set family (transaction database) $F$ is a
 collection of subsets of a set (an itemset) $E$.
The computation time of the algorithm is $O(\sum_{u\in V}
 \sum_{w\in N[u]} |\{ v\in N[w]\ |\ v<u\}|)$.
The term ``$\sum_{u\in V} \sum_{w\in N[u]}$'' means that ``for all pairs of
 $u$ and its closed neighbors $w$'', thus is equivalent to
 $\sum_{w\in V} \sum_{u\in N[w]}$.
Therefore, the computation time is $\sum_{w\in V} \sum_{u\in N[w]}
 |\{ v\in N[w]\ |\ v>u\}| = O(\sum_{w\in V} |N[w]|^2)$.

Because we want to understand the practical performance of this 
 algorithm, we consider about the assumption of Zipf's law.
It is known that real world data often satisfies Zipf's law, 
 for example scale free graphs.
We assume that the degree sequence of the similarity graph is within
 a Zipf's law, i.e., the expected value of the degree of the $i$th
 vertex is $\alpha / i^\Delta$ for some constant $\alpha$ and $\Delta$.
For stating a time complexity, we consider the case in which for some
 constants $c$ and $\beta$, any vertex $w$ satisfies
 $|N[w]| \le \mbox{max}\{ c \alpha / w^\Delta, \beta\}$.
Since the time complexity does not change even if few vertices violate the 
 condition, this assumption would be acceptable.
We have the following by Zipf's law.

\begin{eqnarray*}
 \sum_{w\in V} |N[w]|^2 &\le& \sum_{w\in V} \mbox{max}\{ c \alpha / w^\Delta, \beta\}^2\\
 &\le& \sum_{w\in V} (c \alpha / w^\Delta)^2 + \sum_{w\in V} \beta^2\\
 &\le& \alpha^2 \sum_{w\in V} (1 / w^{2\Delta}) + O(|V|)
\end{eqnarray*}

Thus, the computation time is bounded by $O(\alpha^2\log |V|)$ for
 $\Delta=1/2$, and $O(\alpha^2)$ for $\Delta>1/2$.
For example, if $\alpha$, which is intuitively the number of
 appearances of the most frequently appearing items, is $O(|V|^{1/2})$, 
 the algorithm terminates in $O(|V|\log |V|)$ time if $\Delta=1/2$, and 
 $O(|V|)$ time if $\Delta>1/2$.

\begin{theorem}
For a given similarity graph $G=(V,E)$, algorithm NeighborIntersection
 terminates in $O(\alpha^2\log |V|)$ time for $\Delta=1/2$, and $O(\alpha^2)$
 time for $\Delta>1/2$, if the degree of the $i$th vertex is at most
 $\max\{ c \alpha / i^\Delta, \beta\}$ time, where $c, \beta$, and $\Delta$ 
  are constant numbers.
\end{theorem}

An implementation of this algorithm is available at our Website\\
 ({\tt http://research.nii.ac.jp/\~{}uno/codes.html}).
The name of the implementation is SSPC.

\subsection{Equivalence to frequent itemset mining}

A transaction database is a database composed of records such that 
 each record $T_i$ is a subset of a set of items $I$.
The {\em frequency} of an itemset $P\subseteq V$ is the number of 
 transactions including $P$.
For a threshold $\sigma$, $P$ is called a {\em frequent itemset}
 if its frequency is no less than $\sigma$.
The problem of enumerating all frequent itemsets is a fundamental problem
 in data mining thus has been extensively studied and has many applications.
Consider the transaction database $\{ N[u] | u\in V\}$ where $V$ is 
 regarded as $I$.
Since the number of transactions including itemset $\{ T_u, T_v\}$ is equal to
 the number of vertices $w$ such that $u,v\in N[w]$, 
 the frequency of $\{ T_u, T_v\}$ is equal to $|N[u]\cap N[v]|$.
Thus, we can find all pairs of vertices having non-empty neighbor
 intersection by enumerating all frequent itemsets of size two with $\sigma=1$.
In fact, LCM algorithm\cite{LCM04,DS04} for frequent itemset mining 
 uses an algorithm similar to NeighborIntersection to compute the
 frequencies of itemsets $P\cup \{ e\}$ for all $e$, and it has a
 good practical efficiency for real world large scale databases.

\section{Maximal Clique Enumeration}

Maximal clique enumeration also has long history from
 the 1970's\cite{Ak73,TsIdArSh77}.
In the 2000's, researches began on practical large scale data
 as the growth of data centric science.
The number of maximal cliques is usually tractable even for large
 scale data in such area. 
Recent algorithms such as Makino-Uno algorithm\cite{mace} (MACE)
 and Tomita's\cite{tomita} can enumerate in short time
 in such cases.
The enumeration frameworks of these algorithms are quite different, but 
 both they work well in practice.
As reported in \cite{tomita}, Tomita's algorithm is faster in dense graphs, 
 and MACE is faster in sparse graphs.

MACE is a depth first search type algorithm traversing the maximal 
 clique space.
It starts from the lexicographically maximum clique, and moves to another
 clique recursively by adding a vertex, removing non-adjacent vertices
 to it, and adding addible vertices lexicographically.
Duplication is avoided by introducing a tree shaped transversal route
 induced by parent-child relation between maximal cliques.
Since the algorithm explores no redundant clique, the computation time
 per maximal clique is bounded by polynomial, and the practical
 performance is good.

On the other hand, Tomita's algorithm is a branch and bound type.
It starts from the empty set, and recursively adds vertices
 in a lexicographic order.
At every iteration, it chooses a pivot vertex, and re-orders the
 vertices so that the vertices adjacent to the pivot are located at
 the latter part.
Since any clique composed only of vertices adjacent to the pivot is
 not maximal because we can add the pivot to it, we can omit the
 recursive calls with the addition of these vertices, in this ordering.
This pruning works well especially in dense graphs.
When we want to find only maximal cliques of at least a given threshold size,
 we have one more advantage.
We can delete vertices of degrees smaller than the threshold,
 even in the subgraph given to the recursive calls.
This pruning is not applicable to MACE.

\section{Experiments}

We now explain the results of our computational experiments.
Our implementations were coded in C, without any sophisticated library
 such as binary trees.
All experiments were conducted on a standard PC that are for consumer's use.
Note that none of the implementations used multi-cores.
The codes and the instances are available at our Website
 ({\tt http://research.nii.ac.jp/\~{}uno/codes.html}).
As test instances, we prepared three types of data; randomly generated
 data, business relation among Japanese companies, and Reuters news articles.
As shown below, the results suggest that our data polishing algorithm
 satisfies the requirements for micro-clustering, that are quantity,
 independence, coverage, granularity, and rigidity.

The first instances are randomly generated graphs for evaluating
 how well clustering methods can efficiently find micro-clusters.
First, we initialized a graph with no edge, and randomly put $h$
 original cliques of size 30 on the graph, so that each vertex
 is included in at most $b$ cliques, by re-choosing a vertex when 
 that vertex was already included in $b$ cliques.
We made graphs of 50,000, 100,000, and 350,000 vertices, and
 $h$ is set to 1,000, 3,000, and 10,000 for each case, respectively.
$b$ is multiplicity that the number of micro-clusters one vertex
 can belong to, and in practice it would range from one to five
 or bit more.
For each vertex $v$, we randomly re-connect several edges incident to $v$.
For each edge $(v,u)$, with probability $1-p$, we replace the end point
 $u$ by a randomly selected vertex $u'$.
Note that this operation does not change the degree of $v$ but 
 changes the degrees of $u$ and $u'$.
We tried the graphs of all combinations of $b=1,2,4$,
 $h=1,000, 3,000, 10,000$, and $p=0.3, 0.5, 0.9, 1.0$

 
We evaluated accuracy by F-value-like measure that takes into account
 both precision and recall.
For each original clique $C$, let $\kappa(C)$ be the maximum of
 $|C\cap C'|$ among all clusters $C'$ found by an algorithm, and
 for each $C'$, let $\kappa(C')$ be the maximum of 
 $|C\cap C'|$ among all original cliques.
We can then consider the precision and recall as $\kappa(C) / |C|$ and 
  $\kappa(C') / |C'|$, respectively.
Let $P$ be the average of $\kappa(C) / |C|$ and $R$ be the average of 
 $\kappa(C') / |C'|$ over all original cliques and clusters.
The F-value-like measure is defined by $(P+R) / (2\times P\times R)$.
The reason we choose the best for precision and recall comes 
 from the principle of data mining.
Data mining approaches basically enumerate all candidates.
In this sense, it is preferable if we can find some solutions that
 correspond to hidden structures, even though the number of
 solutions is a bit large.

We examined Metis\cite{Metis}, Girvin-Newman\cite{Newman}, DBscan\cite{DBscan}
   and graph polishing with the Jaccard
 coefficient with thresholds $0.07, 0.15$ and $0.3$.
DBscan inputs sparse distance matrix, thus we set distance of two vertices
 to one if there connected by an edge, and otherwise the distance is 
 not defined.
The parameters ``eps'' is set to 1.5, and ``MinPts'' is set to 30 for
 $b=1$, and to $50$ otherwise.
Actually, the cluster sizes produced by Girvin-Newman were quite biased,
 the accuracy were not good.
The results by Metis and DBscan were also so.
On the contrary, the results by graph polishing were good in many cases.
As the decrease of the similarity threshold, the number of clusters
 and the accuracy increased.
In particular, the accuracy was almost 1 in less noisy cases, 
 while that of Metis was below 0.5.
The number of clusters was not so large, even with threshold 0.07.
The accuracy did not change much when the problem sizes or $b$ increased.
It is interesting that the accuracy of Metis increased in such cases.


The second instance involved business relation data among Japanese
 companies in which the vertices are the companies, and 
 two companies are connected when they trade.
The data is of year 2012 and provided by Teikoku DataBank Limited, Japan.
The numbers of vertices and edges in the graph are 3,282 and 35,168, 
 respectively, and the graph has 32,953 maximal cliques.

We apply our graph polishing algorithm with the PMI
 (pointwise mutual information) of 0.5, 0.6, 0.7, 0.8, 0.9.
The PMI of $A$ and $B$, $A,B\subseteq E$
 is a similarity measure defined by $\log ( |A\cap B|\times |E| /
  |A|\times|B|)$.
The changes in the number of edges and number of maximal cliques
 are summarized in Table \ref{t1}.
The visualization of the graphs are shown in Figure \ref{t2}. 
Many clusters were clarified as cliques, and some 
vertices belonged two or more clusters, while we cannot see anything
 in the original graph.
Examples of clusters are as follows. 

\begin{itemize}
\item Toyota Motor, Suzuki Motor, Yamaha Motor, Daihatsu Motor, Mazda Motor, Isuzu Motor, Nissan Motor, Hino Motor, Fuji Heavy Industries, Honda Motor, Mitsubishi Motor, Sato Shoji, Jidosha Buhin Kogyo
\item Nissan Shatai, Aichi Machine Industry, Calsonic Kansei, Sincerity Passion Kindness, JFE Container, SNT, Zero, Unipres, Kokusai, Gexeed
\item Tsuchiya, Insight, Career Bank, Japan Care Service, Ecomic, JPN Holdings, Accordia Golf
\end{itemize}

\begin{table}
\caption{Results when $b=1$ (acc. = accuracy, \#cls. = \#clusters)}\label{b1}
\small
\vspace{-3mm}
\hspace{-10mm}
\begin{tabular}{|l|rr|rr|rr|rr|rr|rr|}
\hline
 & Polish & 0.07 & Polish & 0.15 & Polish & 0.3 & Metis & & DBscan & & Newman & \\ 
 \#cls., $p$ & acc. & \#cls. & acc. & \#cls. & acc. & \#cls. & acc. & \#cls. & acc. & \#cls. & acc. & \#cls. \\ \hline \hline
$1,000, 0.3$ & 0.993 & 1,601 & 0.798 & 1,552 & 0.229 & 4,891 & 0.072 & 1,396& 0.256 & 710 & 0.329 & 27 \\ \hline
$3,000, 0.3$ & 0.991 & 3,217 & 0.484 & 5,075 & 0.146 & 1,631 & 0.119 & 3,299& 0.190 & 2,404 & 0.158 & 18\\ \hline
$10,000, 0.3$ & 0.993 & 11,088 & 0.574 & 15,970 & 0.153 & 10,169 & 0.305 & 11,431 & 0.203 & 7,868 & 0.140 & 40\\ \hline \hline
$1,000, 0.5$ & 0.999 & 1,155 & 0.998 & 1,026 & 0.619 & 2,772 & 0.065 & 1,396& 0.381 & 732 & 0.312 & 59\\ \hline
$3,000, 0.5$ & 1 & 4,273 & 0.991 & 3,152 & 0.299 & 10,375 & 0.117 & 3,299 & 0.315 & 2,380 & 0.159 & 36\\ \hline
$10,000, 0.5$ & 1 & 15,379 & 0.993 & 10,390 & 0.352 & 37,892 & 0.302 & 11,431 & 0.329 & 7,952 & 0.080 & 87\\ \hline \hline
$1,000, 0.9$ & 1 & 1,173 & 1 & 3,364 & 1 & 2,929 & 0.064 & 1,373 & 0.798 & 785 & 0.451 & 396 \\ \hline
$3,000, 0.9$ & 1 & 4,249 & 1 & 4,300 & 1 & 3,467 & 0.114 & 3,277 & 0.803 & 2,379 & 0.352 & 636 \\ \hline
$10,000, 0.9$ & 1 & 14,585 & 1 & 17,472 & 1 & 13,106 & 0.304 & 11,326 & 0.799 & 7,843 & 0.183 & 1,001 \\ \hline \hline
$1,000, 1.0$ & 1 & 999 & 1 & 999 & 1 & 999 & 0.092 & 1,000 & 1 & 1,000 & 1 & 1,000 \\ \hline
$3,000, 1.0$ & 1 & 2,999 & 1 & 2,999 & 1 & 2,999 & 0.154 & 3,000 & 1 & 3,000 & 1 & 3,000 \\ \hline
$10,000, 1.0$ & 1 & 9,999 & 1 & 9,999 & 1 & 9,999 & 0.634 & 10,000 & 1 & 10,000  & 1 & 10,000 \\ \hline
\end{tabular}
\vspace{-3mm}
\end{table}

\begin{table}
\caption{Results when $b=2$ (acc. = accuracy, \#cls. = \#clusters)}\label{b2}
\small
\vspace{-3mm}
\hspace{-10mm}
\begin{tabular}{|l|rr|rr|rr|rr|rr|rr|}
\hline
 & Polish & 0.07 & Polish & 0.15 & Polish & 0.3 & Metis & & DBscan & & Newman &\\
 \#cls., $p$ & acc. & \#cls. & acc. & \#cls. & acc. & \#cls. & acc. & \#cls. & acc. & \#cls. & acc. & \#cls. \\ \hline \hline
$1,000, 0.3$ & 0.853 & 1,694 & 0.679 & 1,448 & 0.197 & 3,742 & 0.057 & 1,194 & 0.304 & 63 & 0.287 & 17 \\ \hline
$3,000, 0.3$ & 0.518 & 15,868 & 0.463 & 4,252 & 0.145 & 3,871 & 0.111 & 2,888 & 0.269 & 511 & 0.267 & 46\\ \hline
$10,000, 0.3$ & 0.616 & 15,683 & 0.498 & 14,389 & 0.151 & 15,928 & 0.371 & 9,917 & 0.274 & 1,517 & 0.269 & 113\\ \hline \hline
$1,000, 0.5$ & 0.998 & 2,111 & 0.944 & 1,129 & 0.476 & 2,217 & 0.072 & 1,194 & 0.417 & 147 & 0.286 & 17\\ \hline
$3,000, 0.5$ & 0.999 & 4,533 & 0.867 & 3,740 & 0.269 & 6,518 & 0.110 & 2,888 & 0.359 & 697 & 0.291 & 23\\ \hline
$10,000, 0.5$ & 1 & 13,402 & 0.881 & 11,880 & 0.298 & 22,267 & 0.386 & 9,917 & 0.366 & 2,218 & 0.356 & 31\\ \hline \hline
$1,000, 0.9$ & 0.999 & 1,489 & 1 & 4,967 & 0.896 & 5,306 & 0.066 & 1,183 & 0.477 & 270 & 0.451 & 396\\ \hline
$3,000, 0.9$ & 1 & 31,289 & 1 & 8,136 & 0.834 & 9,409 & 0.111 & 2,874 & 0.359 & 785 & 0.161 & 122\\ \hline
$10,000, 0.9$ & 1 & 168,539 & 1 & 26,567 & 0.842 & 33,678 & 0.341 & 9,861 & 0.372 & 2,704 & 0.190 & 195\\ \hline \hline
$1,000, 1.0$ & 0.999 & 1,303 & 0.999 & 1,303 & 0.999 & 1,303 & 0.082 & 781 & 0.459 & 241 & 0.082 & 26\\ \hline
$3,000, 1.0$ & 1 & 3,776 & 1 & 3,776 & 1 & 3,776 & 0.145 & 2,117 & 0.365 & 711 & 0.098 & 23\\ \hline
$10,000, 1.0$ & 1 & 10,758 & 1 & 10,758 & 1 & 10,758 & 0.475 & 7,149 & 0.375 & 2,387 & 0.301 & 43\\ \hline
\end{tabular}
\vspace{-3mm}
\end{table}


\begin{table}
\caption{Results when $b=4$ (acc. = accuracy, \#cls. = \#clusters)}\label{b4}
\small
\vspace{-3mm}
\hspace{-10mm}
\begin{tabular}{|l|rr|rr|rr|rr|rr|rr|}
\hline
 & Polish & 0.07 & Polish & 0.15 & Polish & 0.3 & Metis & & DBscan & & Newman &\\
 \#cls., $p$ & acc. & \#cls. & acc. & \#cls. & acc. & \#cls. & acc. & \#cls. & acc. & \#cls. & acc. & \#cls. \\ \hline \hline
$1,000, 0.3$ & 0.819 & 2,162 & 0.658 & 1,438 & 0.198 & 3,789 & 0.055 & 1,166 & 0.273 & 131 & 0.281 & 20\\ \hline
$3,000, 0.3$ & 0.460 & 23,929 & 0.454 & 4,108 & 0.153 & 4,441 & 0.108 & 2,787 & 0.215 & 692 & 0.236 & 67\\ \hline
$10,000, 0.3$ & 0.543 & 24,965 & 0.483 & 1,4053 & 0.157 & 1,6940 & 0.511 & 9,587 & 0.227 & 2111 & 0.290 & 152\\ \hline \hline
$1,000, 0.5$ & 0.994 & 2,666 & 0.916 & 1,128 & 0.473 & 1,973 & 0.068 & 1,166 & 0.329 & 193 & 0.325 & 16\\ \hline
$3,000, 0.5$ & 0.995 & 5,925 & 0.815 & 3,853 & 0.282 & 5,904 & 0.113 & 2,787 & 0.230 & 690 & 0.271 & 16\\ \hline
$10,000, 0.5$ & 0.997 & 13,653 & 0.832 & 12,449 & 0.304 & 20,301 & 0.532 & 9,587 & 0.247 & 2,201 & 0.377 & 47\\ \hline \hline
$1,000, 0.9$ & 1 & 2453 & 0.997 & 5,761 & 0.866 & 4,830 & 0.067 & 1,156 & 0.363 & 218 & 0.235 & 45\\ \hline
$3,000, 0.9$ & 1 & 14,784 & 0.989 & 10,756 & 0.768 & 9,262 & 0.110 & 2,775 & 0.207 & 592 & 0.172 & 57\\ \hline
$10,000, 0.9$ & 1 & 254,357 & 0.991 & 32,705 & 0.782 & 3,2676 & 0.470 & 9,543 & 0.230 & 2,062 & 0.168 & 118\\ \hline \hline
$1,000, 1.0$ & 1 & 1803 & 0.997 & 1,939 & 0.926 & 1,705 & 0.081 & 754 & 0.371 & 205 & 0.138 & 27\\ \hline
$3,000, 1.0$ & 1 & 5703 & 0.985 & 8,018 & 0.808 & 5,735 & 0.134 & 1,986 & 0.237 & 523 & 0.319 & 37\\ \hline
$10,000, 1.0$ & 1 & 12,343 & 0.988 & 19,293 & 0.826 & 18,565 & 0.681 & 6,739 & 0.260 & 1,838 & 0.251 & 41\\ \hline
\end{tabular}
\vspace{-3mm}
\end{table}

\begin{figure}[t]
  \begin{center}
  \includegraphics[scale=0.45]{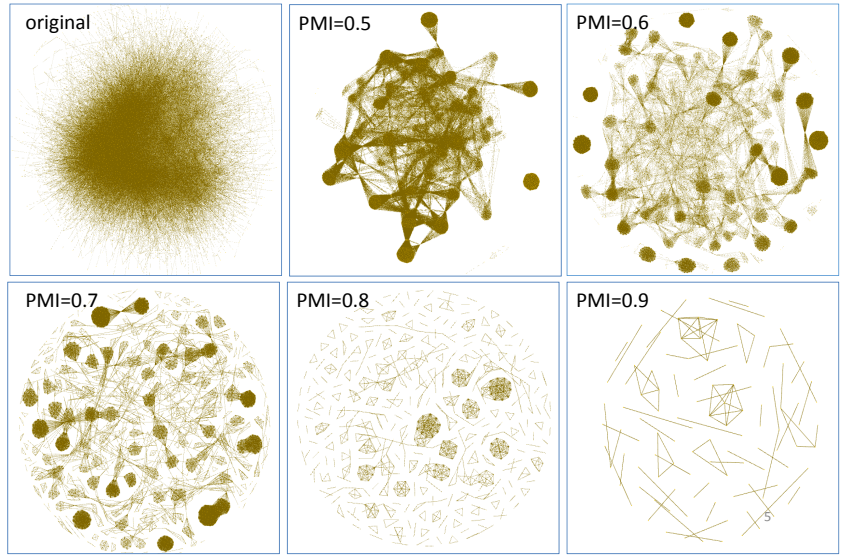}
  \end{center}
  \caption{Visualization of the business relation graph and a polished graph}\label{t2}
\end{figure}

\begin{table}[t]
\begin{center}
\caption{Results on business relation data with different PMI thresholds}\label{t1}
\begin{tabular}{|l|c|c|c|c|c|}
\hline
 PMI & 0.9 & 0.8 & 0.7 & 0.6 & 0.5\\ \hline
 \#edges & 93 & 1,172 & 13,786 & 73,132 & 232,173 \\ \hline
 \#cliques & 59 & 341 & 521 & 343 & 254 \\ \hline
\end{tabular}
\end{center}
\end{table}

The first cluster is of all Japanese car manufacturers, a car parts
 manufacturer, and a trading company of metal materials.
The second is of  car parts manufacturers, trading companies on car parts,
 and those of transportation and containers.
The third is of human resources, land development, and investment.
Four of them are companies from Hokkaido, a prefecture of Japan.
We can see deep relations among the members according to corporate affairs
 and regions.
On the other hand, the clusters seemed to cover the related companies, 
 for example all motor companies are in a cluster.


The third instance is of news articles from Reuters news, a well known
 mass media.
The name of the dataset is RCV1 (English), and the provider is 
 the National Institute of Standards and Technology (NIST)\cite{reuter}.
The articles are from 20/Aug/1996 to 19/Aug/1997, and the number
 of articles is 806791.
Each article was converted to ``bag of words'' which is a set of words
 that the article includes.
Even though two or more identical words appear in an article, we consider
 them as appearing once, thus no two identical words appeared in the bag
 of words of an article.
To remove words appearing quite often, and some special words such as
 ID, we ignored words appearing in at least 1\% of articles, or less 
 than ten articles.

We constructed a similarity graph of the articles so that vertices were
 articles, and two vertices were connected when the Jaccard coefficient
 between corresponding articles was no less than 0.2.
We applied maximal clique enumeration (MACE), Metis, Girvin-Newman, and
 our micro-clustering to the graph and obtained the clusters.
MACE did not terminate in one hour, and produced more than 50
 million maximal cliques. The sizes were relatively large on average,
 likely due to exponentially many cliques included in a large dense
 subgraph.
The number of the clusters in the results of Girvin-Newman, Metis,
 micro-clustering was 25,826, 97,260 and 96,607, respectively.
The sizes of clusters by Girvin-Newman were biased, so that 
 the maximum and second maximum are 240,787 and 75,384, i.e., 
 30\% and 9\% of all the articles.
At the same time there were many small clusters 2 or 3 in size.
The sizes from Metis were not biased; almost all 
 clusters were 7 or 8 in size, and the maximum was 10.
However, it seemed to be too much uniform.
We could find several clusters that would be mixtures of several topics
 such as the following two clusters.

\noindent
TOKYO 1996-11-22 JAPAN: Daimon -96/97 parent forecast.\\
TOKYO 1996-11-22 JAPAN: Daimon - 6mth parent results.\\
TOKYO 1996-08-30 JAPAN: Daimon - 96/97 div forecast.\\
TOKYO 1997-03-28 JAPAN: Daimon -96/97 parent forecast.\\
BRUSSELS, Belgium BELGIUM: EU, Mexico sign cognac and tequila deal.\\
MANILA 1996-10-31 PHILIPPINES: PHILIPPINE STOCKS - factors to watch - Oct 31.\\
BRUSSELS 1996-10-03 BELGIUM: WTO finds against Japan on liquor tax - EU official.\\

\noindent
WASHINGTON 1996-10-21 USA: PRESS DIGEST - Washington Post Business - Oct 22.\\
WASHINGTON 1997-03-27 USA: Farallon group raises Strawbridge stake.\\
OLDWICK, N.J. 1997-04-15 USA: A.M. Best upgrades Exel to A plus from A.\\
LONDON 1997-06-02 UK: EXEL to buy stake in Lloyds managing agency.\\
NEW YORK 1997-06-09 USA: UC'NWIN says appoints Niall Duggan as CEO.\\
PHILADELPHIA 1997-07-14 USA: Strawbridge cuts distribution to holders.\\
AKRON, Ohio 1996-10-01 USA: ABC Dispensing names Crate CFO.\\

In the results of our graph polishing, the cluster sizes range
 from 1 to 292, and most of the clusters were of from 2 to 10 in size.
Not many clusters seemed to be mixtures, as with Metis's.
On the contrary, there were many articles that had the same title,
 or almost the same titles.
Their contexts must be deeply related, thus they should be
 in the same cluster.
We selected some strings composed of several words that were common
 prefixes of many articles, and counted the number of clusters
 including the strings.
The results are summarized in Table \ref{reuter1}. Here (1) means
 that the number of clusters is one, but two clusters of (1) 
  were merged into one.
We can see that Metis tended to partition articles that should 
 have been included in the same cluster.
On the other hand, Girvin-Newman tended to gather non-deeply
 related articles in a cluster.
Our graph polishing did not tend to do both, and seemed
 to obtain good clusters, compared to two popular clustering methods.
In the examples in Table \ref{reuter1}, we showed some clusters that
 would correspond to local areas. 
In such cases, the words used in the article should be categorized
 according to the local areas, such as India, and the result of
 the Girvan-Newman will be good.

\begin{table}[t]
\caption{Number of clusters including specified strings}\label{reuter1}
\begin{tabular}{|l|c|c|c|}
\hline
 & Metis & Girvin-Newman & Polishing \\ \hline
Egyptian pound averages & 18 & 5 & 7 \\ \hline
INDIA: INDIA GOVERNMENT SECURITIES & 37 & 1 & 1 \\ \hline
India Kothari Pioneer MMMF & 4 & 1 & 1 \\ \hline
TAIWAN: Taiwan BSPA & 13 & 1 & 2\\ \hline
USA: High Plains Wheat & 26 & 4 & 4\\ \hline
Eurobonds - Expected new issues -Middle East & 3 & 1 & 1\\ \hline
Eurobonds - Expected new issues -Asia & 1 & 1 & 1\\ \hline
Eurobonds - Expected new issues (-central) & 44 & (1) & 1\\ \hline
Eurobonds - Expected new issues -Latam & 6 & (1) & 1\\ \hline
\end{tabular}
\end{table}

Further, we tried instances taken from twitter, that is of tweets including
 a specified word, and make records of bag of words.
We tried ``lunch'' as the word, and applied clustering algorithms.
Actually the results were not good, and we could see no interesting cluster.
This would be because that the tweets are similar to each other, 
 and topics change smoothly.
For example, there would be records of ``go, restaurant, lunch'', 
``go, tasty, restaurant, lunch'', ``visit, tasty, restaurant, lunch'', 
``visit, tasty, restaurant, lunch, twice'',...
This intuitively implies that the boundaries of the topics are not
 clear, and topics are connected smoothly, in the data of bag of words.
In such data, it seems that we can hardly find good clusters, or 
 there would be no clusters.

\begin{figure}[t]
\vspace{10mm}
  \begin{center}
  \includegraphics[scale=0.45]{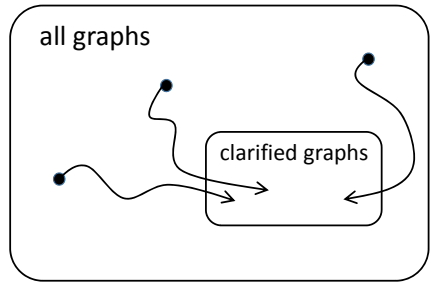}
  \end{center}
  \caption{Unpolished data converted to polished data by data polishing}\label{polishing}
\end{figure}

\section{Discussion}

Data polishing can be seen as to find polished data from given
 un-polished data (see Figure \ref{polishing}).
In the case of graph, we use a feasible hypothesis based algorithm for
 the task.
However, we think another approach, for example minimizing the difference
 between the original graph and polishing graph.
The approach has a theoretical advantage for model explanation, but 
 the minimization problem is a hard optimization requiring long
 computation time.
Although feasible hypothesis approaches cannot guarantee regarding
 approximation, it ensures that the obtained polished
 graph is generated with natural modifications, thus has less
 information loss.

Consider a graph made from a clique by subdividing each edge, where
 a subdivision replaces an edge by a path of two edges.
In some models and data this graph should be recognized as a cluster,
 but our graph polishing deletes all edges, since any vertex pair has
 at most one common neighbor.
On the other hand, consider a chain of cliques that overlap
 neighboring cliques.
Graph polishing changes this graph to a clique if the size of each
 overlap is sufficiently large.
In a usual sense, the cliques should be split if the chain is too
 long.
As we discussed, data polishing is not ideal for taking into account
 global structures.
At this point, we should consider that the clusters provided from graph 
 polishing are still seeds or unification, but are better compared
 to those from existing algorithms.

We can apply clustering algorithms to polished graphs.
In our experiments in another research\cite{nakapara}, the clusters
 of features obtained by these algorithms on polished graphs
 increased the accuracy of a machine learning task.
Polished graphs are not only for clustering, thus there are
 other possible application uses.

Similarity graph construction is a key issue to our graph polishing.
The computational cost is usually large, and sometimes we have 
 no approach to improve efficiency.
In such cases, we can consider the use of approximation.
There are several approaches to find similar elements from the data
 approximately and quickly.
The obtained similarity graph is different from the exact one, but
 feasible hypothesis should also hold in the approximate graph.
This should hold when data and clusters are larger, since 
 the members should have many common neighbors.

\section{Conclusion}

We discussed requirements for the enumeration of  many relatively
 small clusters of elements that are deeply related to each other, 
 that we call micro-clustering.
We proposed a new concept of data processing called ``data polishing'',
 and formulated the clustering problem by data polishing and clique
 enumeration.
Data polishing clarifies the local structures and groups by
 actively modifying the data according to some simple feasible hypothesis.
The design of our algorithm is simple, 
 and thus easily applicable to many kinds of data.
Several statements were proved to ensure that we never miss large 
 clusters in the similarity graph.
We also showed a simple and efficient technique for the data polishing, 
 and showed some complexity results for the computation time of the 
 algorithm under the assumption of the Zipf distribution.

The computational experiments demonstrated the efficiency of our 
 data polishing algorithm regarding accuracy and utility.
Clusters provided with data polishing are usually not split nor
 mixed with others, while clusters often are with existing 
 clustering algorithms.
The sizes, quantities, and granularity of the clusters were good 
 for practice.
In previous studies, the clusters in big data were often used
 to data analysis algorithms such as machine learning,
 image recognition, and prediction.
However, the disadvantages on the utility in the results of existing
 clustering seem to decrease the speed in the practical use.
The development of data polishing would increase the practicality
 of clusters in such areas.

Interesting future work would be to develop data polishing models
 for other kinds of data processing such as bi-clustering, segmentation,
 visualization, flow detection, and anonymization.
It would also be interesting to improve our graph polishing algorithm
 so that we can deal with non-graphic data.\\

\noindent {\bf Acknowledgments}: 
We gratefully thank Professor Takashi Washio of Osaka University, 
Professor Masaru Kitsuregawa of National Institute of Informatics, Professor
 Koji Tsuda and Masashi Sugiyama of Tokyo University, for their valuable
 comments, information, and support.
We also express our appreciation to Teikoku DataBank Limited, Japan, 
 for supplying the Japanese company business relation data.
This research is supported by JST CREST, Japan.

\end{document}